\documentclass{article}
\usepackage{natbib}
\usepackage{amssymb}
\usepackage{amsmath}
\newtheorem{definition}{Definition}
\newtheorem{proposition}{Proposition}
\newenvironment{proof}{\vspace{-1mm}\noindent{\bf {Proof.}}\ }{\hfill $\Box$\\}

\begin{document}

\title{Can rational choice guide us to correct {\em de se}
  beliefs?\thanks{This paper appears in {\em Synthese}, Volume 192, Issue
    12, pp.~4107-4119, December 2015.  The final publication is available
    at Springer via http://dx.doi.org/10.1007/s11229-015-0737-x}}
\author{Vincent Conitzer\\Duke University} 
\date{}

\maketitle

\begin{abstract}
  Significant controversy remains about what constitute correct
  self-locating beliefs in scenarios such as the Sleeping Beauty problem,
  with proponents on both the ``halfer'' and ``thirder'' sides.  To attempt
  to settle the issue, one natural approach consists in creating decision
  variants of the problem, determining what actions the various candidate
  beliefs prescribe, and assessing whether these actions are reasonable
  when we step back.  Dutch book arguments are a special case of this
  approach, but other Sleeping Beauty games have also been constructed to
  make similar points.  Building on a recent article (James R.~Shaw. {\em
    De se} belief and rational choice. {\em Synthese}, 190(3):491-508,
  2013), I show that in general we should be wary of such arguments,
  because unintuitive actions may result for reasons that are unrelated to
  the beliefs.  On the other hand, I show that, when we restrict our
  attention to {\em additive} games, then a thirder will necessarily
  maximize her {\em ex ante} expected payout, but a halfer in some cases
  will not (assuming causal decision theory).
  I conclude that this does not necessarily settle the issue and speculate about what might.\\
\noindent {\bf Keywords:} Sleeping Beauty, Dutch books, decision theory, game theory.
\end{abstract}

\section{Introduction}

The {\em Sleeping Beauty} problem~\citep{Elga00:Self} illustrates some
fundamental issues regarding self-locating beliefs.  In it, a study
participant referred to as ``Sleeping Beauty'' is put to sleep on Sunday,
and awoken either just on Monday, or on both Monday and Tuesday, according
to the outcome of a fair coin toss (Heads or Tails, respectively).  After
an awakening, she is put back to sleep {\em with her memory of the
  awakening event erased}, so that all awakenings are indistinguishable to
her.  When Beauty is awoken, what should be her credence (subjective
probability) that the coin came up Heads?  Some (``halfers'') argue that it
should be $1/2$.  The standard argument for this position is that this
should have been her credence in Heads before the experiment, and she has
learned nothing new, knowing all along that she would be awoken at least
once.  Others (``thirders'') argue that it should be $1/3$.  The standard
argument for this position is that if the experiment is repeated many
times, in the long run, only $1/3$ of awakenings correspond to a toss of
Heads.  (It should be emphasized that ``halfers'' and ``thirders'' would
compute other fractions on different examples, and ``halfing'' and
``thirding'' are supposed to refer to the methods of computing these
fractions rather than these specific values.)  For a summary of reasons why
philosophers are interested in the Sleeping Beauty problem,
see~\cite{Titelbaum13:Ten}.

One approach to settling what Beauty ought to believe is to design
scenarios where she must act on her beliefs, and to investigate the
consequences of being a thirder or a halfer on these actions.  One specific
line of attack within this general approach is to design {\em Dutch book}
arguments.  A Dutch book is a set of bets that an agent would all adopt
individually in spite of the fact that their combination will lead to a
guaranteed loss. If such can be constructed, this is an argument against
the rationality of the agent's beliefs.  In the context of the Sleeping
Beauty problem, the focus is on {\em diachronic} Dutch books, which involve
bets at different times.  Dutch book arguments for the Sleeping Beauty
problem are considered
by~\cite{Hitchcock04:Beauty},~\cite{Halpern06:Sleeping},~\cite{Draper08:Diachronic},~\cite{Briggs10:Putting},
and~\cite{Conitzer15:Dutch}.  These arguments generally favor thirding,
though it is sometimes also argued that a halfer can resist Dutch books,
particularly when adopting evidential decision theory.  \cite{Shaw13:De}
more generally pursues the agenda of integrating {\em de se} beliefs with
rational choice in the context of variants of the Sleeping Beauty problem.
He allows Beauty to play more complex games, and designs one where, he
argues, the thirder makes the wrong decision and the halfer makes the right
decision, regardless of whether they adopt causal or evidential decision
theory.

In this article, taking~\cite{Shaw13:De} as a starting point, I further
pursue the agenda of settling the correct answer to the Sleeping Beauty
problem by looking at the consequences of halfing and thirding on the
outcomes of associated decision problems.  I first sound a note of caution
by showing that in some cases unintuitive outcomes in these examples result
not from incorrect credences, but rather from challenges that a rational
actor faces when trying to coordinate with her past and future selves under
imperfect recall (at least under causal decision theory).  From examples
that involve such challenges, we cannot comfortably draw any conclusions
about the (in)correctness of a particular approach for computing credences.
Subsequently, I show that if we restrict the types of decision problem to
{\em additive} ones, which include typical Dutch book arguments, these
coordination challenges disappear; moreover, under causal decision theory,
a thirder will always make decisions that maximize her overall expected
payout, but a halfer in some cases does not.  I conclude by assessing how
much we can learn from these results about correct self-locating beliefs.

\section{Review: Shaw's Waking Game}

First, a review of Shaw's Waking Game is in order. He argues that thirders
get the wrong answer in this game while halfers get it right.  I focus here
on his analysis of a thirder who is a causal decision
theorist.\footnote{Throughout, unless otherwise noted, I will focus on
  causal decision theory.  Therefore, some of the conclusions I reach can
  be avoided by dismissing causal decision theory.  If the reader feels
  compelled to do so by the examples provided here, then that might be an
  even more significant impact for them to have---but I myself am not
  willing to go that far.}

{\bf Shaw's Waking Game.}  At the beginning of the experiment, Beauty is
informed of the rules of the game, which are as follows.  A fair coin will
be tossed; the outcome of this coin toss will not be revealed to Beauty
until the game is over. If it lands Heads, she will be woken up only once,
on Monday.  If it lands Tails, she will be woken up four times, on Monday,
Tuesday, Wednesday, and Thursday.  Each day, she will be asked to press
either Left of Right.  Her memory of the awakening will be erased
afterwards, she will not be able to take any notes, and the awakenings will
be indistinguishable.  She will be compensated as follows.
 \begin{enumerate}
 \item If Heads came up and she pressed Left, she will receive \$400.
 \item If Heads came up and she pressed Right, she will receive \$200.
 \item If Tails came up and she pressed Left on each of the four days, she
   will receive \$100.
 \item If Tails came up and she pressed Right on each of the four days, she
   will receive \$200.
 \item If Tails came up and she pressed Left on Monday and Right on at
   least one other day, she will receive \$200.
 \item If Tails came up and she pressed Right on Monday and Left on at
   least one other day, she will receive \$100.
 \end{enumerate}

 Shaw makes two assumptions that he calls {\em Randomizing Prohibited} and
 {\em Previous Runs}.  The meaning of the former is clear; the latter
 refers to the fact that Beauty, having seen many similar experiments
 performed on others, has become convinced that a subject always makes the
 same decision on each of her awakenings.  These imply the following, which
 is all that is needed for his analysis of the case of a thirder who is a
 causal decision theorist.

\begin{definition}
  Beauty is said to accept {\em Consistency in Other Rounds} if, upon any
  given awakening, she does not assign any credence to the following event:
  she has woken up or will wake up (with the same information) multiple
  additional times and did not/will not take the same action on each of
  those other occasions.
\end{definition}

Then, Shaw provides the following analysis.  If Beauty is a thirder and a
causal decision theorist, then upon an awakening, she should assign $1/5$
credence to Heads/Monday, $1/5$ to Tails/Monday, and $3/5$ to Tails/some
other day.  If she accepts {\em Consistency in Other Rounds}, then moreover
she believes that either (a) on all other awakenings (if any) she chooses
Left, or that (b) on all other awakenings she chooses Right.  Under (a), if
she chooses to now press Left, her expected payout will be
$$(1/5) \cdot \$400 + (1/5) \cdot \$100 + (3/5) \cdot \$100 = \$160$$
On the other hand, if she chooses to now press Right, her expected payout will be
$$(1/5) \cdot \$200 + (1/5) \cdot \$100 + (3/5) \cdot \$200 = \$180$$
Hence, under (a), she is better off pressing Right.

Under (b), if she chooses to now press Left, her expected payout will be
$$(1/5) \cdot \$400 + (1/5) \cdot \$200 + (3/5) \cdot \$100 = \$180$$
On the other hand, if she chooses to now press Right, her expected payout
will be
$$(1/5) \cdot \$200 + (1/5) \cdot \$200 + (3/5) \cdot \$200 = \$200$$
Hence, under (b), she is also better off pressing Right!  It follows that
Beauty, if she is a thirder and a causal decision theorist, will press
Right.

Now, because all awakenings are indistinguishable, she should {\em always}
press Right, resulting in a payout of \$200.  But always pressing Left
would have resulted in an expected value of \$250, which is better
(assuming Beauty is risk-neutral), and is hence the correct course of
action according to Shaw.  (He shows that a thirder who is an evidential
decision theorist also should choose Right in this example, but I will not
review this analysis here.)

\section{Three Awakenings}
\label{se:three}

Shaw's Waking Game is illuminating, but I believe little can be concluded
from it about whether thirding or halfing is correct.  To show why, let us
consider another example that shares key features of the reasoning above,
but without any coin tosses whatsoever.

{\bf Three awakenings.}  At the beginning of the experiment, Beauty is
informed of the rules of the game, which are as follows.  She will be woken
up exactly three times (Monday, Tuesday, and Wednesday).  Each day, she
will be asked to press either Left of Right.  Her memory of the awakening
will be erased afterwards, she will not be able to take any notes, and the
awakenings will be indistinguishable.  She will be compensated as follows.
 \begin{enumerate}
 \item If she never pressed Right, she will receive \$200.
 \item If she pressed Right once, she will receive \$300.
 \item If she pressed Right twice, she will receive \$0.
 \item If she pressed Right three times, she will receive \$100.
 \end{enumerate}
 Again, note that no coins are tossed at all in Three
 Awakenings.\footnote{In this sense, it is closer to the example of O'Leary
   awakening twice in his trunk~\citep{Stalnaker81:Indexical}, except that
   I need three rather than two awakenings.  Nevertheless, I will stick
   with the Beauty terminology for expository purposes, and will
   reintroduce coin tosses soon.}  The only uncertainties that Beauty faces
 are (1) which day it is and (2) what she herself has done and will do on
 the other days.  In fact, arguably, (1) does not even matter because in
 this game, all awakenings are treated symmetrically.  The key uncertainty
 is (2).

 How should Beauty act in this game?  If she always presses Left, she will
 obtain \$200; if she always presses Right, she will obtain only \$100.  So
 something is to be said for pressing Left.  However, upon any given
 awakening, Beauty can reason as follows.  There are two other rounds in
 which she has pressed or will press a button.  If she accepts {\em
   Consistency in Other Rounds}, then she believes that either (a) she has
 pressed or will press Left both other times or (b) she has pressed or will
 press Right both other times.  In case (a), she will be better off
 pressing Right this round, because pressing Right in only one round pays
 out \$300, whereas never pressing Right pays out \$200.  In case (b), she
 will {\em also} be better off pressing Right this round, because pressing
 Right in all three rounds pays out \$100, whereas pressing Right in only
 two pays out nothing.  So {\em in either case} Beauty is better off
 pressing Right, gaining \$100 from doing so!\footnote{Of course, to reason
   this way, Beauty must be a causal decision theorist; if she were an
   evidential decision theorist, then she would prefer to press Left and
   therefore believe that she presses Left in the other rounds as well.
   The example may thus provide some ammunition for evidential decision
   theorists, but again, I will attempt to steer clear of that debate here
   as much as possible.}  Then, because all awakenings are
 indistinguishable, it seems we should expect Beauty to press Right all the
 time---even though pressing Left all the time results in a higher payout.

 From Three Awakenings, it becomes clear that, under causal decision
 theory, actions that are locally optimal---at least when assuming {\em
   Consistency in Other Rounds}---can result in globally suboptimal
 outcomes, even in cases where there is no ambiguity about what the correct
 credences are.  (I take it to be uncontroversial that Beauty's credence
 upon awakening should be distributed uniformly ($1/3$, $1/3$, $1/3$)
 across Monday, Tuesday, and Wednesday.)  I believe the example also makes
 it clear that the total payout earned by a subject is a very unreliable
 indicator of the correctness of her credences.\footnote{One might, of
   course, argue that this is so only because we are using causal decision
   theory and causal decision theory is flawed.  Still, given the
   prominence of causal decision theory, I believe the example should leave
   us generally cautious about the strategy of using rational choice to
   determine what the correct credences are.}  To drive home the point,
 consider the following modification of Three Awakenings.
 
 {\bf Three Awakenings with a Coin Toss.}  The experiment now begins with a
 biased coin toss.  If it lands Heads (which happens 99\% of the time), we
 proceed with the original Three Awakenings game.  If it lands Tails (1\%),
 Beauty will similarly be woken up on Monday, Tuesday, and Wednesday, and
 asked to press Left or Right, but the payoffs will be different.  In fact,
 they will be much simpler: she will receive \$100 for each time she
 presses Left (and nothing for pressing Right).  As always, Beauty knows
 the setup of this modified game, but will not receive any evidence of how
 the coin landed until the game has ended.

 I take it to be uncontroversial that upon any awakening, Beauty should
 place a credence of 99\% on the event that the coin landed Heads, because
 whether the coin landed Heads or not, she will be awoken three times.
 Moreover, in all six possible awakening events, she will have the exact
 same information.  Given this, the modification is too slight to have an
 impact on her decision: for any given awakening, there is a 99\% chance
 that she will gain \$100 from pressing Right (assuming {\em Consistency in
   Other Rounds}) and a 1\% chance that she will lose \$100 from doing
 so---so she should still press Right.  But now, suppose that Beauty's
 credence is inexplicably inverted, so that she believes that there is a
 99\% chance that the coin came up {\em Tails}. If so, then from her
 perspective, now the simpler payoff function dominates and clearly she
 should press Left.  As a result, she will actually obtain a {\em larger}
 expected payout from the actual game, because always pressing Left results
 in a higher payout in Three Awakenings than always pressing Right.
 However, it seems clear that this should not lead us to believe that
 Beauty's inverted credence is in any sense {\em correct}; rather, she was
 just {\em lucky} that she accidentally inverted the credences, thereby
 escaping the detrimental reasoning to which understanding the game
 correctly would have led her.

 Of course, we do not need to go to such lengths to find examples where
 incorrect credences lead to a better result.  Someone who for some reason
 believes that in roulette Red comes up $2/3$ of the time, and bets on Red
 once for this reason only (as opposed to not betting at all), may well get
 lucky on that one spin of the wheel.  If so, nobody will argue that this
 {\em ex post} outcome implies that the credence of $2/3$ was correct.
 What is interesting about Three Awakenings with a Coin Toss is that any
 credences that maximize {\em ex ante} expected payoff are clearly
 incorrect.  It would seem that it is a very reasonable criterion for
 evaluating the correctness of credences to see whether they lead to the
 maximum {\em ex ante} expected payoff---but the example shows that this
 approach is, in general, problematic (at least if we are not willing to
 dismiss causal decision theory).

\section{Additive Games}

The examples in Section~\ref{se:three} suggest that in sufficiently rich
decision variants of the Sleeping Beauty problem, under causal decision
theory, the payouts that Beauty obtains do not provide useful guidance for
what her correct credences should be.  This is so because in such
scenarios, actions that are locally apparently rational may lead to
suboptimal payouts even when there can be no serious dispute about what the
correct credences should be.  But perhaps, if we restrict the space of
scenarios, we can avoid such issues.

The problematic aspect of the Three Awakenings game is that Beauty's
``three selves'' need to {\em coordinate} their actions to maximize
payout---the effect of one action on overall payout depends on the other
actions---and they fail to do so due to the lack of memory.  What happens
if we assume away this interdependence?  In what follows, I show that in
the resulting restricted class of games---additive games---Beauty does in
fact maximize her expected total payout by being a thirder (and a causal
decision theorist).  Of course, merely showing an example additive game
where being a thirder maximizes Beauty's expected total payout will do
little to prove the point, because for all we know there is another example
where being a thirder results in suboptimal payout.  I have to prove the
result at some level of generality for it to be more than merely
suggestive.  In particular, for the sake of generality, I wish to allow
that Beauty does not necessarily have the same experience in each awakening
(thereby allowing us to also address examples such as ``Technicolor
Beauty''~\citep{Titelbaum08:The}).  To do so, I will have to be a bit more
formal.

\begin{definition}
\label{def:additive}
A Sleeping Beauty decision variant with payoff function $\pi$ is {\em
  additive} if for every realization $r$ of the initial coin
toss,\footnote{We may assume without loss of generality that a single coin
  toss at the beginning provides all the randomness needed for the duration
  of the game, since we can keep as much of this randomness hidden from
  Beauty as we must, for as long as we must.  Indeed, it is commonly agreed
  that moving the coin toss between Sunday night and Monday night in the
  standard version of the Sleeping Beauty problem makes no difference.}
\begin{itemize}
\item {\bf (actions do not affect future rounds)} $r$ always leads to the same
  number $n_r$ of awakenings by Beauty regardless of Beauty's actions, and
  for every $i$ with $1 \leq i \leq n_r$, the information that Beauty
  possesses in the $i$th awakening depends only on $r$ and $i$, and not on
  Beauty's earlier actions; and
\item {\bf (payoff additivity)} for every $i$ with $1 \leq i \leq n_r$, and
  every two corresponding sequences of actions $a_1, \ldots, a_{n_r}$ and
  $a'_1, \ldots, a'_{n_r}$ that Beauty may take upon her $n_r$ awakenings,
  we have that $$\pi(r, a_1, \ldots, a_{n_r}) - \pi(r, a_1, \ldots, a_{i-1},
  a'_i, a_{i+1}, \ldots a_{n_r}) =$$ $$\pi(r, a'_1, \ldots, a'_{i-1}, a_i,
  a'_{i+1}, \ldots a'_{n_r}) - \pi(r, a'_1, \ldots, a'_{n_r})$$
\end{itemize}	
\end{definition}

Intuitively, in additive games, Beauty does not need to worry about
coordinating her actions with her selves from other awakenings.  This is
because by the first condition, the only effect of actions is directly on
the final payout (as opposed to them affecting the number of awakenings or
the information that she has in future rounds), and by the second condition
these effects on payout are independent across actions.  This intuition
leads to the following proposition.

\begin{proposition}
  If Beauty is a thirder and a causal decision theorist, and acts
  accordingly upon each individual awakening, then she will maximize her
  {\em ex ante} expected payout in additive games.  If she is a halfer and
  a causal decision theorist, and acts accordingly on each individual
  awakening, there are additive games in which she does not maximize her
  {\em ex ante} expected payout.
\label{prop:additive}
\end{proposition}

\begin{proof}
  For each $r$ and $i$ with $1 \leq i \leq n_r$, let $v(r,i)$ correspond to
  the awakening event on the $i$th day after a coin toss realization of
  $r$.  Let $V = \bigcup_{(r,i): 1 \leq i \leq n_r} \{v(r,i)\}$ be the set
  of all awakening events.  By payoff additivity, we can construct, for
  every $v \in V$, a function $\pi_v$ such that Beauty's total payout upon
  coin toss realization $r$ and actions $a_1, \ldots, a_{n_r}$ is
  $c(r) + \sum_{i \in \{1,\ldots,n_r\}} \pi_{v(r,i)}(a_i)$, where $c(r)$ is
  a constant that we may ignore for the purpose of acting
  optimally.\footnote{To be specific, we can choose, for every $v$, a
    default action $d_v$.  Let $r(v)$ denote the coin toss realization that
    leads to $v$.  Then, for any action $a_v$ that can be taken at $v$, we
    let
    $\pi_v(a_v) = \pi(r(v), d_{v(r,1)}, \ldots, d_{v(r,i-1)}, a_v,
    d_{v(r,i+1)}, \ldots, d_{v(r,n_r)}) - \pi(r(v), d_{v(r,1)}, \ldots,
    d_{v(r,i-1)}, d_{v}, d_{v(r,i+1)}, \ldots, d_{v(r,n_r)})$,
    where $v=v(r,i)$.  By payoff additivity it then follows that
    $\pi(r, a_1, \ldots, a_{n_r}) = \pi_{v(r,1)}(a_1) + \pi(r, d_{v(r,1)},
    a_2, \ldots, a_{n_r}) = \pi_{v(r,1)}(a_1) + \pi_{v(r,2)}(a_2) + \pi(r,
    d_{v(r,1)}, d_{v(r,2)}, a_3, \ldots, a_{n_r}) = \ldots = (\sum_{i \in
      \{1,\ldots,n_r\}} \pi_{v(r,i)}(a_i)) + \pi(r(v), d_{v(r,1)}, \ldots,
    d_{v(r,n_r)})$,
    so we can set $c(r)=\pi(r(v), d_{v(r,1)}, \ldots, d_{v(r,n_r)})$. (It
    is easy to see that conversely the existence of such $\pi_v(\cdot)$
    implies payoff additivity.)}  We will use $I \subseteq V$ to denote an
  {\em information set}, i.e., a set of awakening events that Beauty cannot
  distinguish.\footnote{Note that one awakening event corresponds to many
    nodes in the standard extensive-form representation of the game---one
    for each sequence of actions that Beauty has taken so far.  However,
    because of the ``actions do not affect future rounds'' condition, all
    these nodes must lie in the same information set.}  Note that two
  awakening events in the same information set may correspond either to the
  same coin toss realization $r$---e.g., subsequent Monday and Tuesday
  awakenings in the standard version of Sleeping Beauty---or to different
  coin toss realizations---e.g., the two Monday awakening events
  corresponding to Heads and Tails in the standard version.  When Beauty
  awakens in information set $I$, if she is a thirder, then her credence
  that the realization of the coin toss is $r$ is given by
  $P(r | I) = \frac{P(r) \cdot \nu(I,r)}{\sum_{r'}P(r') \cdot \nu(I,r')}$,
  where $\nu(I,r) = |\{v \in I: r(v) =r\}|$ is the number of times that
  Beauty will awaken with information $I$ after coin toss realization $r$
  and $r(v)$ is the realization that leads to $v$.  (This is the essence of
  being a thirder: given particular information upon awakening, credence in
  a particular realization is proportional to the number of times one will
  awaken with this information under this realization.  Indeed, if the
  experiment is repeated many times, then $P(r | I)$ gives the long-run
  fraction of the awakenings in information set $I$ that corresponded to a
  coin toss realization of $r$.)  Moreover, the credence that she assigns
  to a specific $v \in I$ with $r(v)=r$ is
  $P(v|I)=\frac{P(r|I)}{\nu(I,r)} = \frac{P(r)}{\sum_{r'}P(r') \cdot
    \nu(I,r')}$.
  Hence, if $A_I$ is the set of actions available to her in information set
  $I$,\footnote{Note that an agent cannot have different sets of actions
    available to her in two awakening events that are in the same
    information set, because then she would be able to rule out some of the
    awakening events in the information set based on the actions available
    to her.  Some Dutch book arguments are flawed because they violate this
    criterion.} she will choose some $a_I \in A_I$ that maximizes
  $\sum_{v \in I} P(v|I) \pi_v(a_I)$.

  If Beauty takes action $a_I \in A_I$ whenever she is in information set
  $I$, then her {\em ex ante} expected payout overall is
  $\sum_r P(r) \sum_{i \in \{1,\ldots,n_r\}} \pi_{v(r,i)}(a_{I(v(r,i))})$
  (where $I(v)$ is the information set in which $v$ lies).  Rearranging,
  this is equal to $\sum_I \sum_{v \in I} P(r(v)) \pi_v(a_I)$.\footnote{To
    see this, note that the first summation sums over all $v$ by first
    summing over all $r$ and then over all $v$ corresponding to that $r$.
    The second summation also sums over all $v$, but instead by first
    summing over all information sets and then over all $v$ in that
    information set.  In both cases, the summand for $v$ is
    $P(r(v))\pi_v(a_{I(v)})$.}  We will show that if Beauty is a thirder
  and a causal decision theorist, then in fact for every $I$ she maximizes
  $\sum_{v \in I} P(r(v)) \pi_v(a_I)$, thereby establishing that she
  maximizes her {\em ex ante} expected payout overall.  Indeed, we have
  already established that for each $I$, Beauty maximizes
  $\sum_{v \in I} P(v|I) \pi_v(a_I)$.  Using that (Beauty being a thirder)
  $P(v|I) = \frac{P(r(v))}{\sum_{r'}P(r') \cdot \nu(I,r')}$, we obtain that
  Beauty maximizes
  $\frac{\sum_{v \in I} P(r(v)) \pi_v(a_I)}{\sum_{r'}P(r') \cdot
    \nu(I,r')}$.
  Because Beauty cannot affect the denominator of this expression, this is
  equivalent to maximizing $\sum_{v \in I} P(r(v)) \pi_v(a_I)$, as was to
  be shown.

  On the other hand, if Beauty is a halfer (and a causal decision
  theorist), then consider the standard Sleeping Beauty game, where a coin
  is tossed to determine whether she awakens once (upon Heads) or twice,
  with all her three possible awakenings in the same information set. Let
  her choose between Left and Right upon each awakening.  If the awakening
  is one corresponding to Heads, she will receive $3$ for choosing Left
  (and $0$ for Right); if it is one corresponding to Tails, she will
  receive $2$ for choosing Right (and $0$ for Left).  If Beauty is a halfer
  (and a causal decision theorist), upon awakening she will think it
  equally likely that she is in a Heads awakening and that she is in a
  Tails awakening, and therefore will choose Left for an expected payoff of
  $3/2$ in this round (rather than Right for $1$).  However, overall,
  choosing Right every time gives an {\em ex ante} expected total payout of
  $(1/2) \cdot 2 \cdot 2 = 2$, whereas choosing Left every time gives an
  {\em ex ante} expected total payout of $(1/2) \cdot 1 \cdot 3 = 3/2$, so
  Beauty fails to maximize her expected payout.
\end{proof}
Intuitively, the way the proof works is as follows.  Because the game is
additive, we can separate Beauty's total {\em ex ante} expected payoff into
the contributions made to it by individual information sets $I$.  It then
remains to show that Beauty maximizes her expected payoff for each
information set $I$ if she is a thirder and a causal decision theorist.
Now, the contribution of each individual awakening event $v$ within the
information set $I$ to the expected payoff is proportional to the
probability $P(r(v))$ of the coin toss realization $r(v)$ that gives rise
to $v$.  But, when she is in $I$, Beauty's credence $P(v|I)$ in $v$ is {\em
  also} proportional to $P(r(v))$. This is so because (being a thirder) her
credence $P(r(v)|I)$ in $r(v)$ is proportional to $P(r(v))\nu(I,r(v))$,
where $\nu(I,r(v))$ is the number of awakening events in $I$, across which
this credence is equally divided.  Because of this, Beauty weighs the
awakening events in an information set exactly so as to maximize {\em ex
  ante} expected payoff.  In contrast, if she is a halfer and a causal
decision theorist, then her credence in $r(v)$ is not proportional to
$P(r(v))\nu(I,r(v))$ but rather just to $P(r(v))$, so that her credence in
$v$ itself is proportional to $P(r(v)) / \nu(I,r(v))$.\footnote{At least,
  it would appear natural to split the credence equally across these
  $\nu(I,r(v))$ awakening events---but note that the counterexample does
  not actually rely on this.}  As a result, she places too little weight on
awakening events $v$ in $I$ that correspond to coin toss outcomes $r$ that
lead to many other awakening events in $I$.  This is what leads her to
decide suboptimally in the counterexample at the end of the proof: she
insufficiently weighs the Tails awakenings in making her
decision.\footnote{\label{fo:evidential}One may wonder whether, along the
  lines of~\cite{Briggs10:Putting}, the halfer could correct for this by
  adopting evidential decision theory instead.  The idea would be that her
  decision provides evidence for what she does in {\em all} the
  $\nu(I,r(v))$ awakenings, thereby undoing the problematic division by
  $\nu(I,r(v))$ above.  Unfortunately, if she adopts evidential decision
  theory, then in general her decision will also provide evidence about
  what she does in {\em other} information sets (especially, very similar
  ones) and this prevents the proof from going through.  To illustrate,
  consider the following example (an additive game).  We toss a three-sided
  coin (Heads, Tails, and Edge with probability 1/3 each).  On Heads,
  Beauty will be awakened once in information set $I_1$; on Tails, once in
  information set $I_2$; on Edge, once in $I_1$ {\bf and} once in $I_2$. On
  every awakening, Beauty must choose Left or Right.  If the world is Heads
  or Tails, Left pays out 3 and Right 0; if it's Edge, Left pays out 0 and
  Right 2.  Note that $I_1$ and $I_2$ are completely symmetric.  The
  optimal thing to do from the perspective of {\em ex ante} expected payout
  is to always play Left (and get $(2/3) \cdot 3$ rather than
  $(1/3) \cdot 2 \cdot 2$ from Right).  What will the EDT halfer do?  Upon
  awakening in (say) $I_1$, she will assign credence $1/2$ to each of Heads
  and Edge (and $0$ to Tails).  (In fact, some variants of halfing will
  result in different credences; to address such a variant, we can modify
  the example by adding another awakening in both Heads and Tails---but not
  Edge---worlds, in an information set $I_3$ where no action is taken.  All
  variants of halfing---and, for that matter, thirding---of which I am
  aware will result in the desired credences of $1/2$ Heads, $1/2$ Edge in
  this modified example.)  Now, the key point is that if she plays Right
  (Left) now, this is very strong evidence that she would play Right (Left)
  in $I_2$ as well---after all the situation is entirely symmetric.  Thus,
  conditional on playing Left, she will expect to get $3$ in the Heads
  world and $0$ in the Edge world; conditional on playing Right, she will
  expect to get $0$ in the Heads world and $2 \cdot 2=4$ in the Edge world.
  Hence she will choose Right (and by symmetry she will also choose Right
  in $I_2$), which does not maximize {\em ex ante} expected payout.
  \cite{Conitzer15:Dutch} provides a more elaborate example along these
  lines in the form of a Dutch book to which evidential decision theorists
  fall prey, along with further discussion.  (Incidentally, an evidential
  decision theorist who is a {\em thirder} fails to maximize {\em ex ante}
  expected payoff on a much simpler example: in the counterexample at the
  end of the proof of Proposition~\ref{prop:additive}, just change the
  payoff for choosing Left on Heads to $5$.  Now Left maximizes {\em ex
    ante} expected payoff, but an evidential decision theorist who is a
  thirder will calculate $(1/3)\cdot 5 = 5/3 < 8/3 = (2/3)\cdot 2 \cdot 2$
  and choose Right.  What goes wrong is that $\nu(I,\text{Tails})=2$ now
  occurs {\em twice} on the right-hand side, once due to thirding ($2/3$)
  and once due to evidential decision theory (the second $2$; the third $2$
  is the payoff for choosing Right on Tails). I thank an anonymous reviewer
  for providing this counterexample.  It should also be noted
  that~\cite{Briggs10:Putting} already gives a Dutch book for an evidential
  decision theorist who is a thirder.)}

Proposition~\ref{prop:additive} also implies that Beauty, if she is a
thirder and a causal decision theorist, is invulnerable to certain types of
Dutch books.  (This is already discussed in prior
work~\citep{Hitchcock04:Beauty,Draper08:Diachronic,Briggs10:Putting}.)
Specifically, she will not fall for a Dutch book as long as: (a) Beauty, at
the beginning of the experiment, is made aware of the bets she will be
offered in different awakening states and will not forget this; (b)
Beauty's betting actions affect neither her future awakening states nor the
outcomes of past or future bets; (c) for every two states in the same
information set, the bet posed to Beauty is the same.  Here, (c) seems
natural, because if two states in the same information set were to have
different bets associated with them, then in fact, by (a), they would allow
Beauty to distinguish between them before she takes her action,
contradicting that they are in the same information set.  (a) and (c)
together ensure that we can interpret the bets as Beauty playing a game
(whose rules she knows), and adding (b) ensures that this game is additive.
(Note that we may have to add an initial round to correspond to a bet at
the beginning of the experiment.)  By the first part of
Proposition~\ref{prop:additive}, Beauty will act in a way that maximizes
her expected payout.  This means she cannot be vulnerable to accepting a
set of bets that results in a sure loss, because if she did so she would
not be maximizing her expected payout (since, after all, she can also
accept none of the bets at all and thereby avoid a loss).  Given all this,
the second part of Proposition~\ref{prop:additive} is unsurprising in light
of the Dutch book given by~\cite{Hitchcock04:Beauty} for halfers that use
causal decision theory.

One may wonder whether additive games are really the ``right'' class of
games to which to restrict our attention.  Perhaps the result can be
generalized to a somewhat broader class of games, for example by slightly
relaxing the first condition in Definition~\ref{def:additive}.\footnote{It
  should be noted that doing so appears nontrivial.  For example, suppose
  we continue to insist that the number of awakenings depends only on the
  outcome of the coin toss, but we attempt to relax the requirement that
  actions do not affect the information that Beauty has in future
  awakenings.  Then, an action's value may come less from the payoff
  resulting directly from it and more from allowing Beauty to obtain
  increased payoffs in later rounds by improving her information.  It is
  possible that these latter, indirect effects on payoffs are not additive
  even when the direct payoffs are additive (so that payoff additivity is
  technically satisfied), and that this would still allow us to embed
  problematic examples such as the Three Awakenings game.}  Such a
generalization, of course, would only strengthen the point.  More
problematically, perhaps a different natural class of games would actually
favor halfing.  I cannot rule out this possibility, but it seems unlikely
to me that such a class would be more compelling than that of additive
games.  I believe that additive games are well motivated by the discussion
given at the beginning of this section about removing the coordination
problem between Beauty's multiple selves, and the fact that the result
provides a corollary about Dutch books is also encouraging.

\section{Conclusion}

What can we conclude from the foregoing?  First and foremost, the Three
Awakenings game shows that we should be very cautious when drawing
conclusions about halfing vs.~thirding from the outcomes of
decision-theoretic variants of the Sleeping Beauty problem.  I do believe
that Proposition~\ref{prop:additive} shows some merit to being a thirder
rather than a halfer, but surely it does not settle the matter once and for
all.  One might well argue, for example, that, once she has awakened under
particular circumstances, Beauty should no longer care whether she
maximizes her {\em ex ante} expected payout; instead, she should maximize
her expected payout with respect to her beliefs at hand.  These two
objectives turn out to be aligned in the case of a (causal decision
theorist) thirder in additive games, and this may be a nice property.  But
the battle-hardened halfer is likely more comfortable biting the bullet and
accepting nonalignment in these two objectives than giving up on other
cherished philosophical commitments.  Another possibility for the halfer
may be to embrace a version of evidential decision theory instead.  More
discussion of how halfers may or may not benefit from adopting evidential
decision theory, particularly in the context of Dutch book arguments, is
given
by~\cite{Arntzenius02:Reflections},~\cite{Draper08:Diachronic},~\cite{Briggs10:Putting},
and~\cite{Conitzer15:Dutch} (see also the discussion in
Footnote~\ref{fo:evidential}).

How could we create decision variants of the Sleeping Beauty problem that
leave no ambiguity about whether rational decisions truly correspond to
rational beliefs?  One way to do so would be to consider a {\em myopic}
Beauty.  Such a Beauty would be rewarded immediately after taking an action
in the game, rather than at the end.  We may suppose that she is rewarded
in something giving immediate satisfaction---say, chocolate---rather than
money.  Moreover, she is assumed to care only about the very near future;
tomorrow is too far in the future to affect her decisions.  Her being
myopic is not to be understood as her being irrational. We still assume her
to be entirely rational, but she just discounts the future exceptionally
heavily (and, to the extent it matters, the past as well).  Such a Beauty,
in a simple variant (without decisions) where she is certainly woken up on
both Monday and Tuesday but given chocolate only on Tuesday, will hope that
{\em today is Tuesday} when she is awoken.\footnote{Perhaps such examples
  are more palatable when we consider variants of the Sleeping Beauty
  problem that involve clones---see, e.g.,~\cite{Elga04:Defeating}
  and~\cite{Schwarz14:Belief}.  The example where she hopes that {\em today
    is Tuesday} then is analogous to the ``After the Train Crash'' case
  in~\cite{Hare07:Self}, where a victim of a train crash who has forgotten
  his name, upon learning that the victim named ``A'' will have to undergo
  painful surgery, hopes that {\em he is victim ``B''}.  (See
  also~\cite{Hare09:On,Hare10:Realism}.)}  So a myopic Beauty's preferences
are entirely {\em de se} and {\em de nunc}.  If we additionally suppose
that the game is additive as described above, then she need not worry at
all about what she will do or has done in other rounds (including about
what her current actions say about what she will do or has done in other
rounds), because none of those affect her current circumstances and
rewards.  Hence, it seems that here beliefs and actions should
unambiguously line up.  Unfortunately, such extreme assumptions also make
it difficult, and perhaps impossible, to set up an example that provides
much insight beyond non-decision-theoretic variants of the Sleeping Beauty
problem.  There is a tightrope to be walked here.  Too permissive a setup
will allow us to reach conclusions that are unwarranted; too restricted a
setup will not allow us to reach any conclusions at all.  Perhaps the best
we can hope for is to identify the happy medium and gradually accumulate
bits of evidence that, while each not entirely convincing on its own,
gradually tilt the balance in favor of one or the other
position.\footnote{Not all of these bits of evidence would concern decision
  variants, especially as surprising connections from the Sleeping Beauty
  problem to other problems continue to be drawn.  For
  example,~\cite{Pittard15:When} makes an interesting connection to
  epistemic implications of disagreement that provides a challenge to
  halfers (and argues that this challenge can be met).  Of course, there
  are also many direct probabilistic arguments. Many of these were already
  made early on in the debate about Sleeping
  Beauty~\citep[etc.]{Elga00:Self,Lewis01:Sleeping,Arntzenius02:Reflections,Dorr02:Sleeping},
  but new ones continue to be
  made~\citep[e.g.]{Titelbaum12:Embarrassment,Conitzer14:Devastating}.}

\section*{Acknowledgments}

I thank the anonymous reviewers for many useful comments that have helped
to significantly improve the paper.

\bibliography{beauty}

\begin{thebibliography}{21}
\expandafter\ifx\csname natexlab\endcsname\relax\def\natexlab#1{#1}\fi
\expandafter\ifx\csname url\endcsname\relax
  \def\url#1{{\tt #1}}\fi

\bibitem[Arntzenius(2002)]{Arntzenius02:Reflections}
Frank Arntzenius.
\newblock {Reflections on Sleeping Beauty}.
\newblock {\em Analysis}, 62\penalty0 (1):\penalty0 53--62, 2002.

\bibitem[Briggs(2010)]{Briggs10:Putting}
Rachael Briggs.
\newblock {Putting a value on Beauty}.
\newblock In {Tamar Szab\'o Gendler and John Hawthorne}, editors, {\em {Oxford
  Studies in Epistemology: Volume 3}}, pages 3--34. {Oxford University Press},
  2010.

\bibitem[Conitzer(2014)]{Conitzer14:Devastating}
Vincent Conitzer.
\newblock {A devastating example for the Halfer Rule}.
\newblock {\em Philosophical Studies}, 2014.
\newblock DOI 10.1007/s11098-014-0384-y.

\bibitem[Conitzer(2015)]{Conitzer15:Dutch}
Vincent Conitzer.
\newblock {A Dutch book against sleeping beauties who are evidential decision
  theorists}.
\newblock {\em Synthese}, 2015.
\newblock DOI 10.1007/s11229-015-0691-7.

\bibitem[Dorr(2002)]{Dorr02:Sleeping}
Cian Dorr.
\newblock {Sleeping Beauty: in defence of Elga}.
\newblock {\em Analysis}, 62\penalty0 (4):\penalty0 292--296, 2002.

\bibitem[Draper and Pust(2008)]{Draper08:Diachronic}
Kai Draper and Joel Pust.
\newblock {Diachronic Dutch Books and Sleeping Beauty}.
\newblock {\em Synthese}, 164\penalty0 (2):\penalty0 281--287, 2008.

\bibitem[Elga(2000)]{Elga00:Self}
Adam Elga.
\newblock {Self-locating belief and the Sleeping Beauty problem}.
\newblock {\em Analysis}, 60\penalty0 (2):\penalty0 143--147, 2000.

\bibitem[Elga(2004)]{Elga04:Defeating}
Adam Elga.
\newblock {Defeating Dr.~Evil with self-locating belief}.
\newblock {\em Philosophy and Phenomenological Research}, 69\penalty0
  (2):\penalty0 383--396, 2004.

\bibitem[Halpern(2006)]{Halpern06:Sleeping}
Joseph~Y. Halpern.
\newblock {Sleeping Beauty reconsidered: Conditioning and reflection in
  asynchronous systems}.
\newblock In {Tamar Szab\'o Gendler and John Hawthorne}, editors, {\em {Oxford
  Studies in Epistemology: Volume 3}}, pages 111--142. {Oxford University
  Press}, 2006.

\bibitem[Hare(2007)]{Hare07:Self}
Caspar Hare.
\newblock {Self-Bias, Time-Bias, and the Metaphysics of Self and Time}.
\newblock {\em The Journal of Philosophy}, 104\penalty0 (7):\penalty0 350--373,
  July 2007.

\bibitem[Hare(2009)]{Hare09:On}
Caspar Hare.
\newblock {\em {On Myself, And Other, Less Important Subjects}}.
\newblock {Princeton University Press}, 2009.

\bibitem[Hare(2010)]{Hare10:Realism}
Caspar Hare.
\newblock {Realism About Tense and Perspective}.
\newblock {\em Philosophy Compass}, 5\penalty0 (9):\penalty0 760--769, 2010.

\bibitem[Hitchcock(2004)]{Hitchcock04:Beauty}
Christopher Hitchcock.
\newblock {Beauty and the bets}.
\newblock {\em Synthese}, 139\penalty0 (3):\penalty0 405--420, 2004.

\bibitem[Lewis(2001)]{Lewis01:Sleeping}
David Lewis.
\newblock {Sleeping Beauty: reply to Elga}.
\newblock {\em Analysis}, 61\penalty0 (3):\penalty0 171--176, 2001.

\bibitem[Pittard(2015)]{Pittard15:When}
John Pittard.
\newblock {When Beauties disagree: Why halfers should affirm robust
  perspectivalim}.
\newblock In {Tamar Szab\'o Gendler and John Hawthorne}, editors, {\em {Oxford
  Studies in Epistemology: Volume 5}}. {Oxford University Press}, 2015.
\newblock Forthcoming.

\bibitem[Schwarz(2014)]{Schwarz14:Belief}
Wolfgang Schwarz.
\newblock {Belief update across fission}.
\newblock {\em British Journal for the Philosophy of Science}, 2014.
\newblock Forthcoming.

\bibitem[Shaw(2013)]{Shaw13:De}
James~R. Shaw.
\newblock {{\em De se} belief and rational choice}.
\newblock {\em Synthese}, 190\penalty0 (3):\penalty0 491--508, 2013.

\bibitem[Stalnaker(1981)]{Stalnaker81:Indexical}
Robert~C. Stalnaker.
\newblock {Indexical belief}.
\newblock {\em Synthese}, 49:\penalty0 129--151, 1981.

\bibitem[Titelbaum(2008)]{Titelbaum08:The}
Michael~G. Titelbaum.
\newblock {The relevance of self-locating beliefs}.
\newblock {\em Philosophical Review}, 117\penalty0 (4):\penalty0 555--605,
  2008.

\bibitem[Titelbaum(2012)]{Titelbaum12:Embarrassment}
Michael~G. Titelbaum.
\newblock {An Embarrassment for Double-Halfers}.
\newblock {\em Thought}, 1\penalty0 (2):\penalty0 146--151, 2012.

\bibitem[Titelbaum(2013)]{Titelbaum13:Ten}
Michael~G. Titelbaum.
\newblock {Ten reasons to care about the Sleeping Beauty problem}.
\newblock {\em Philosophy Compass}, 8\penalty0 (11):\penalty0 1003--1017, 2013.

\end{thebibliography}
\bibliographystyle{plainnat}

\end{document}